\documentclass[11pt]{article}

\usepackage{latexsym}
\usepackage{amsmath, amssymb} %
\usepackage{xspace}
\usepackage[usenames,dvipsnames]{color}
\usepackage{color}
\usepackage{xspace}
\usepackage{epic,eepic}
\usepackage{graphicx}

\newtheorem{xdefinition}{Definition}
\newtheorem{xobservation}{Observation}
\newtheorem{xtheorem}{Theorem}
\newtheorem{xmaintheorem}{Main Theorem}
\newtheorem{xlemma}{Lemma}
\newtheorem{xproposition}{Proposition}
\newtheorem{xcorollary}{Corollary}

\newtheorem{xclaim}{Claim}
\newtheorem{xproperty}{Property}
\newenvironment{definition}{\begin{xdefinition}\rm}%
{\hspace*{\fill}\raisebox{-1pt}{\boldmath$\Box$}\end{xdefinition}}
{\hspace*{\fill}\raisebox{-1pt}{\boldmath$\Box$}\end{xobservation}}
\newenvironment{theorem}{\begin{xtheorem}\rm}{\end{xtheorem}}

\newenvironment{lemma}{\begin{xlemma}\rm}{\end{xlemma}}
\newenvironment{proposition}{\begin{xproposition}\rm}{\end{xproposition}}
\newenvironment{corollary}{\begin{xcorollary}\rm}{\end{xcorollary}}

\newenvironment{property}{\begin{xproperty}}{\end{xproperty}}
\newenvironment{proof}{\begin{trivlist}\item[]{\bf Proof }}%
{\hspace*{\fill}\raisebox{-1pt}{\boldmath$\Box$}\end{trivlist}}

\newcommand{\G}{\ensuremath{\mathcal{G}}}
\newcommand{\I}{\ensuremath{\mathcal{I}}}
\newcommand{\ALG}[1]{{\ensuremath{\mathbb{#1}}}\xspace}
\newcommand{\LRU}{\ensuremath{\mathrm{LRU}}\xspace}
\newcommand{\FIFO}{\ensuremath{\mathrm{FIFO}}\xspace}
\newcommand{\FWF}{\ensuremath{\mathrm{FWF}}\xspace}
\newcommand{\LFD}{\ensuremath{\mathrm{LFD}}\xspace}
\newcommand{\SET}[1]{\{ #1 \}}
\newcommand{\SETOF}[2]{\{#1\:|\:#2\}}

\newcommand{\BIGSETOF}[2]{\left\{ #1\:|\:#2 \right\}}
\newcommand{\WEHAVE}{\!:\;}
\newcommand{\SEQ}[1]{\langle #1 \rangle}
\newcommand{\WORST}[3]{\operatorname{Worst}(#1,#2,#3)}
 
\newcommand{\RWOR}[2]{\ensuremath{\mathrm{WR}_{#1,#2}}\xspace}

\begin{document}

\title{Access Graphs Results for LRU versus FIFO  \\
       under Relative Worst Order Analysis\thanks{A
preliminary version of this paper will appear in the
proceedings of the
Thirteenth Scandinavian Symposium and Workshops on Algorithm Theory.
Partially
supported by the Danish Council for Independent Research.} }

\author{Joan Boyar \hspace{2em} Sushmita Gupta  \hspace{2em} Kim S. Larsen \\[1ex]
        University of Southern Denmark \\
        Odense, Denmark \\[1ex]
        {\tt \{joan,sgupta,kslarsen\}@imada.sdu.dk}}

\date{}
\maketitle

\begin{abstract}
Access graphs, which have been used previously in connection with competitive
analysis to model locality of reference in paging, are considered in
connection with relative worst order analysis. 
In this model, \FWF is shown
to be strictly worse than both \LRU and \FIFO on any access graph.
\LRU is shown to be strictly better than \FIFO on paths and cycles, but
they are incomparable on some families of graphs which grow with the
length of the sequences.
\end{abstract}

\section{Introduction}

The term {\em online} algorithm~\cite{BE97b} is used for an algorithm that receives its input as a sequence of items, one at a time, and
for every item, before knowing the subsequent items, must make an irrevocable decision regarding how to process the current item.

The most standard measure of quality of an online algorithm
is {\em competitive analysis}~\cite{G66j,ST85j,KMRS88j}.
This is basically the worst case ratio between the performance
of the online algorithm compared to an optimal offline
algorithm which
is allowed to know the entire input sequence before processing
it and is assumed to have unlimited computational power.

Though this measure is very useful and has driven a lot of research,
researchers also observed problems~\cite{ST85j}
with this measure from the very beginning:
many algorithms
obtain the same (poor) ratio, while showing quite different
behavior in practice.

The {\em paging problem} is one of the prime examples of these difficulties.
The paging problem is the problem of maintaining a subset
of a potentially very large number of pages in a much smaller,
faster cache with space for a limited set of $k$ pages.
Whenever a page is requested, it must be brought into cache if it is
not already there. In order to make room for such a page, another page
currently in cache must be evicted. Therefore, an online algorithm for
this problem is often referred to as an eviction strategy.

For a number of years, researchers have worked on refinements
or additions to competitive analysis with the aim of obtaining
separations between different algorithms for solving an
online problem.
Some of the most obvious and well-known
paging algorithms are the eviction strategies \LRU
(Least-Recently-Used) and \FIFO (First-In/First-Out).
One particularly notable result has been the separation of
\LRU and \FIFO via {\em access graphs}.
Access graphs were introduced in~\cite{BIRS95} with the aim
of modelling the locality of reference that is often seen
in real-life paging situations~\cite{D68,D80}.
An access graph is an undirected graph
with all pages in slow memory as vertices.
Given such a graph, one then restricts the analysis of
the performance of an algorithm to sequences
respecting the graph, in the sense that any two distinct, consecutive
requests must be neighbors in the graph.
Important results in understanding why \LRU is often observed to
perform better than \FIFO in practice
were obtained in~\cite{BIRS95,CN99},
showing that on some access graphs, \LRU is strictly better than
\FIFO, and on no access graph is it worse;
all these previous results are with respect to competitive analysis.

More recently, researchers have made attempts to introduce new
generally-applicable performance measures
and to apply measures defined to solve one particular problem
more generally to other online problems.
A collection of alternative performance measures is
surveyed in~\cite{DLO05j}.
Of the alternatives to competitive analysis,
{\em relative worst order analysis}~\cite{BF07,BFL07j} and
{\em extra resource analysis}~\cite{KP00}
are the ones that
have been successfully applied to most different online problems.
See~\cite{EKL11p} for an example list of online problems and references to
relative worst order analysis results resolving various issues
that are problematic with regards to competitive analysis.

Paging has been investigated
under relative worst order analysis in~\cite{BFL07j}.
Some separations were found, but \LRU and \FIFO were proven
equivalent, possibly because locality of reference is necessary to
separate these two paging algorithms.
In this paper, we apply the access graph technique to
relative worst order analysis.
Note that the unrestricted analysis in~\cite{BFL07j} corresponds
to considering a complete access graph.

Overall, our contributions are the following.
Using relative worst order analysis,
we confirm the competitive analysis result~\cite{BIRS95}
that \LRU is better than \FIFO for path access graphs. Since these
two quality measures are so different, this is a
a strong indicator of the robustness of the result.
Then we analyze cycle access graphs, and show that with regards to
relative worst order analysis, \LRU is strictly better than \FIFO.
Note that this does not hold under competitive analysis.
The main technical contribution is the proof showing that on cycles,
with regards to relative worst order analysis, \FIFO is never
better than \LRU.
Clearly, paths and cycles are the two most fundamental building blocks,
and future detailed analyses of any other graphs type will likely
build on these results.
In addition, when the cache size is small compared with the size of the access
graph, localized behavior in time is likely to be that of paths and cycles.

The standard example of a very 
bad algorithm with the same competitive ratio as \LRU and \FIFO is
\FWF, which is shown to be strictly worse than both \LRU and \FIFO on
any access graph (containing a path of length at least $k+1$),
according to relative worst order analysis. 

Using relative worst order analysis, one can often obtain more
nuanced results. This is also the case here for general access graphs,
where we establish an incomparability result.

None of the algorithms we consider require prior knowledge of the
underlying access graph. This issue was pointed out in \cite{FM97}
and \cite{FR97} in connection with the 
limitations of some of the access graph results given
in \cite{BIRS95,FK95,IKP96j} and the Markov paging analogs in \cite{KPR00j}.

As relative worst order analysis is getting more established
as a method for analyzing online algorithms in general,
it is getting increasingly important that the theoretical toolbox
is extended to match the options available when carrying out competitive
analysis. Recently, in~\cite{EKL11p}, list factoring~\cite{AvSW95,BM85} was
added as an analytical tool when using relative worst order analysis
on list accessing problems~\cite{ST85j,AW98}, and here we demonstrate
that access graphs can be included as another useful technique.

After a preliminary section, where we define all concepts, including
relative worst order analysis, we prove that \LRU is never worse
than \FIFO  on paths or cycles. Then we establish separation results,
showing that \LRU is strictly better than
\FIFO on paths and cycles of length at least $k+1$ and that both
algorithms are strictly better than \FWF on any graph containing a
path of length at least $k+1$.
The last result proves the incomparability of \LRU and \FIFO
on general access graphs, using a family of graphs where the size is
proportional to the length of the request sequence. We conclude with some
open problems regarding determining completely for which classes of graphs
\LRU is better than \FIFO.

\section{Preliminaries}
\label{section-preliminaries}

The {\em paging problem} is the problem of processing a sequence
of page requests with the aim of minimizing the number of page faults.
Pages reside in a large memory of size $N$, but whenever a page
is requested, it must also be in the smaller cache of size $k<N$.
If it is already present, we refer to this as a {\em hit}.
Otherwise, we have a {\em fault} and must bring the page into cache.
Except for start-up situations with a cache that is not full,
this implies that some page currently in cache must be chosen to be evicted
by a paging algorithm.

If \ALG{A} is a paging algorithm and $I$ an input sequence,
we let $\ALG{A}(I)$ denote the number of faults that \ALG{A} incurs on $I$.
This is also referred to as the {\em cost} of \ALG{A} on $I$.

An important property of some paging algorithms that is used 
several times in this paper is the following:

\begin{definition}
An online paging algorithm is called {\em conservative}
if it incurs at most $k$ page faults on any consecutive subsequence
of the input containing $k$ or fewer distinct page references.
\end{definition}

The algorithms, Least-Recently-Used (\LRU) and First-In/First-Out
 (\FIFO) are examples of conservative algorithms.
On a page fault, \LRU evicts the least recently used page in cache and
 \FIFO evicts the page which has been in cache the longest.
Flush-When-Full (\FWF), which is not conservative, is the algorithm
 which evicts all pages in cache whenever there is
a page fault and its cache is full.

Longest-Forward-Distance (\LFD), which is not online,
evicts the page whose next request is the latest.
If there is more than one page which is never requested again,
then any of those pages can be evicted,
and all of these versions of \LFD are optimal~\cite{Bel66}.

An input sequence of page requests is denoted
$I=\SEQ{r_1, r_2, \ldots, r_{|I|}}$.
We use standard mathematical interval notation to denote subsequences.
They can be open, closed, or semi-open, and are denoted by 
$(r_a, r_b)$, $[r_a, r_b]$, $(r_a, r_b]$, or $[r_a, r_b)$.
If $S$ is a set of pages, we call a request interval $S$-{\em free}
if the interval does not contain requests to any elements of $S$.

We use the following notation for graphs.
\begin{definition}
The path graph on $N$ vertices is denoted $P_N$
and a cycle graph on $N$ vertices is denoted $C_N$.
A \textit{walk} is an ordered sequence of vertices
where consecutive vertices are either identical or
adjacent in the graph.
A \textit{path} is 
a walk in which every vertex appears at most once.
The length of a walk ${\cal W}$ is the number of (not necessarily distinct) 
vertices in it, denoted by $|{\cal W}|$.
The set of distinct vertices in a walk ${\cal W}$ is denoted by $\SET{{\cal W}}$.
\end{definition}

\begin{definition}
An {\em access graph} $G=(V,E)$ is a graph whose vertex set corresponds
to the set of pages that can be requested in a sequence. A sequence is
said to {\em respect} an access graph,
if the sequence of requests constitutes a walk in that access graph. 
\end{definition}

In the relative worst order analyses carried out in this paper,
permutations play a key role. We introduce some notation for this
and then present the standard definition of the relative worst order 
quality measure.

For an algorithm $\ALG{A}$, $\ALG{A}_W(I)$ is the cost of the algorithm $\ALG{A}$
on the worst reordering of the input sequence $I$, i.e., $\ALG{A}_W(I) = \max_{\sigma} \ALG{A}(\sigma(I))$, 
where $\sigma$ is a permutation on $|I|$ elements
and $\sigma(I)$ is a reordering of the sequence $I$.

\begin{definition}
For any pair of paging algorithms $\ALG{A}$ and $\ALG{B}$, we define 
\begin{align*}
c_l(\ALG{A}, \ALG{B}) \; = \; &
  \sup \SET{c \mid \exists b\WEHAVE \forall I\WEHAVE \ALG{A}_W(I) \geq c\, \ALG{B}_W(I) -b}
\textrm{ and }\\
 c_u(\ALG{A}, \ALG{B}) \; = \; &
  \inf \SET{c \mid \exists b\WEHAVE \forall I\WEHAVE \ALG{A}_W(I) \leq c\, \ALG{B}_W(I) +b}.
\end{align*}

If $c_l(\ALG{A},\ALG{B}) \geq 1 $ or $c_u(\ALG{A},\ALG{B}) \leq 1$,
the algorithms are said to be {\em comparable}
and the {\em relative worst order ratio} $\RWOR{\ALG{A}}{\ALG{B}}$ of 
algorithm $\ALG{A}$ to $\ALG{B}$ is defined. Otherwise, $\RWOR{\ALG{A}}{\ALG{B}}$ is undefined.     
 \begin{align*}
 &\textrm{If } c_l(\ALG{A}, \ALG{B}) \geq 1, \textrm{ then } \RWOR{\ALG{A}}{\ALG{B}} = c_u(\ALG{A},\ALG{B})
 \textrm{ and} \\
 &\textrm{if } c_u(\ALG{A}, \ALG{B}) \leq 1, \textrm{ then } \RWOR{\ALG{A}}{\ALG{B}} = c_l(\ALG{A},\ALG{B}). 
 \end{align*}

If $\RWOR{\ALG{A}}{\ALG{B}} < 1$, algorithms $\ALG{A}$ and $\ALG{B}$ are said to be comparable
in $\ALG{A}$'s {\em favor}. Similarly, if $\RWOR{\ALG{A}}{\ALG{B}} > 1$,
the algorithms are said to be comparable in $\ALG{B}$'s {\em favor}.
\end{definition}

When we use this measure to compare algorithms on a given
access graph $G$, we use the notation
$\ALG{A}_W^G(I)$ to denote the cost of $\ALG{A}$ on a worst permutation
of $I$ that respects $G$.
Similarly,
we use $\RWOR{\ALG{A}}{\ALG{B}}^G$ to denote the relative worst order ratio
of algorithms $\ALG{A}$ and $\ALG{B}$ on the access graph $G$. 

Finally,
let $\WORST{I}{G}{\ALG{A}}$ denote the set of worst orderings
for the algorithm $\ALG{A}$ of $I$ respecting the access graph $G$,
i.e., any sequence in $\WORST{I}{G}{\ALG{A}}$ is a permutation of $I$,
they all respect $G$, and %
for any $I\in\WORST{I}{G}{\ALG{A}}$, $\ALG{A}(I)=\ALG{A}_W^G(I)$.

\section{Paths}

In~\cite[Theorem~13]{BIRS95}, it has been shown that if the access graph
is a tree, then \LRU is optimal among all online algorithms.
In the case of path graphs, though, \LRU matches the performance of an
optimal offline algorithm. For completeness, we provide our own direct proof. %

\begin{theorem}
\label{lru-opt}
On a path access graph, \LRU's performance is optimal.
\end{theorem}

\begin{proof}
We compare the behavior of \LRU to that of \LFD
on a sequence respecting a path access graph.

When more than one of the pages in cache will not be requested again,
\LFD can arbitrarily choose to evict any of these pages
when bringing a new page into cache.
Without loss of generality, we assume that we compare \LRU to
a version of \LFD that, if \LRU evicts a page which is never
requested again, evicts the same page as \LRU.

Assume to the contrary that there exists a sequence
$I=\SEQ{r_1, \ldots, r_n}$
for which \LFD does strictly better than \LRU.
Both algorithms start with an empty cache and
until the cache is full, they behave identically.
Let $r_i$ be the first request where the algorithms behave differently,
i.e., to bring in the new page, they evict different pages from their caches.

We denote the page requested at $r_i$ by $p$,
and the pages evicted by \LRU and \LFD by $q$ and $\hat{q}$,
respectively.
If neither $q$ nor $\hat{q}$ are requested again, by the assumption
of \LFD version above, \LRU and \LFD should have evicted the same page.
Thus, we may assume that $q$ is requested again after $r_i$.
Since \LRU does not evict $\hat{q}$,
$\hat{q}$ must have been requested more recently than $q$.
Let $r_a$ and $r_b$ denote the last requests before $r_i$
for $q$ and $\hat{q}$, respectively.
It follows from \LFD's eviction strategy
that unless $\hat{q}$ is never requested again,
the first request for $q$ after $r_i$ must be
before the first request for $\hat{q}$ after $r_i$.

By definition of $q$ and $\hat{q}$,
the intervals $(r_a, r)$ and $(r_b, r)$ are $\SET{q}$-free
and $\SET{q,\hat{q}}$-free, respectively.
The request sequence must have the following structure.
\[\ldots\ldots r_a = q \ldots\ldots r_b = \hat{q}
 \underbrace{\ldots\ldots}_{\SET{q,\hat{q}}-\textrm{free}} r_i = p 
 \underbrace{\ldots\ldots}_{\SET{q,\hat{q}}-\textrm{free}} r_c = q 
 \ldots\ldots
\]
It is easy to see that $p$ does not lie on the path $(q, \hat{q})$, since
otherwise $p$ would be requested in $(r_a, r_b)$
and therefore should not be evicted by 
\LRU before evicting $q$ at $r_i$.
Due to the subwalks that are $\SET{q,\hat{q}}$-free,
there is a path from $p$ to $q$ which does not pass through $\hat{q}$, as 
well as a path from $p$ to $\hat{q}$ which does not pass through $q$.

Thus, for the three vertices $p$, $q$, and $\hat{q}$ in the access graph,
we have argued
that none of them are on the path between the two others.
This implies that the access graph is not a path, and we have
reached a contradiction.
\end{proof}

\begin{theorem}
\label{paths-leq}
For all sequences $I$ respecting the access graph $P_N$,
\[\LRU^{P_N}_W(I) \leq \FIFO^{P_N}_W(I).\]
\end{theorem}
\begin{proof}
Consider any sequence $I$ respecting $P_N$. Let $I'$ be a worst ordering
for \LRU among the permutations of $I$ respecting $P_N$. Then, 
$\LRU_W^{P_N}(I) = \LRU(I') \leq \FIFO(I') \leq \FIFO_W^{P_N}(I)$ 
where the first inequality follows from Theorem~\ref{lru-opt}.  %
\end{proof}

\section{Cycles}
\label{all-graphs}

Almost this entire section is leading up to a proof
that for all $I$ respecting the access graph $C_N$,
$\LRU^{C_N}_W(I) \leq \FIFO^{C_N}_W(I)$.

Notice that this theorem is not trivial, since there
exist sequences respecting the cycle access graph
where \FIFO does better than \LRU.
Consider, for example, the cycle on four vertices $C_4=\SEQ{ 1,2,3,
4}$, $k=3$, and the request sequence $I =\SEQ{ 2,1,2,3,4,1 }$. With
this sequence, at the request to $4$, \LRU evicts $1$ and \FIFO evicts
$2$. Thus, \FIFO does not fault on the last request and has one
fault fewer than \LRU. Note that on the reordering, $I' =\SEQ{ 1,2,2,3,4,1 }$,
\LRU still faults five times, but \FIFO does too. This is the transformation
which would be performed in Lemma~\ref{lemma-first-walk-long} below, combined with
the operation in the proof of Lemma~\ref{lemma-reduce} to reinsert 
requests which have been removed. Note that this is not
a worst ordering for \LRU, since \LRU and \FIFO both fault six times
on $I'' =\SEQ{ 1,2,3,4,1,2 }$.

Each of the results leading up to the main theorem in this section
is aimed at
 establishing a new property that we may assume in the rest
of the section.
Formally, these results state that if we can prove our end goal
{\em with} the new assumption, then we can also prove it without.
Thus, it is just a formally correct way of phrasing that we are
reducing the problem to a simpler one.
Some of the sequence transformations we perform in establishing
these properties also remove requests, in addition to possibly
reordering. The following general lemma allows us to do this in all of these
specific cases.

\begin{lemma}
\label{lemma-reduce}
Assume we are given an access graph $G$, a sequence $I$
respecting $G$, and a sequence $I_{\LRU}\in\WORST{I}{G}{\LRU}$.
We write $I_{\LRU}$ as the concatenation of three
subsequences $\SEQ{I_1,I_2,I_3}$.
Let $I'$ be $\SEQ{I_1,I_2',I_3}$, where $I_2'$ can be any
subsequence (not necessarily of the same length as $I_2$)
such that $I'$ still respects $G$.
Assume that \LRU incurs at least as many faults on $I_2'$ as on $I_2$,
and the cache content, including information concerning which
pages are least recently used, is exactly the same just after
$I_2'$ in $I'$ as after $I_2$ in $I_{\LRU}$.
Assume further that
$I_2'$ is obtained from $I_2$ by removing some requests and/or
reordering requests, and that $\SET{I}=\SET{I'}$.  
Then, $I'\in\WORST{I'}{G}{\LRU}$, and if $\LRU(I')\leq\FIFO^G_W(I')$, then
$\LRU^G_W(I)\leq\FIFO^G_W(I)$.
\end{lemma}

\begin{proof}
Since we have not reduced the number of faults
and the state of the cache is unaffected,
$\LRU(I_{\LRU})\leq\LRU(I')$.
If we assume for the sake of contradiction that $I'\notin\WORST{I'}{G}{\LRU}$,
then one would be able to choose a worse ordering $I_C'$,
i.e., with $\LRU(I_C')>\LRU(I')$.
We now create a sequence $I_C$ by
inserting the pages we removed from $I_{\LRU}$ compared with $I'$ into $I_C'$.
We do this by inserting any request to $p$ immediately after an
existing request to $p$ in $I_C'$.
By assumption, these pages all still have requests,
so this is indeed possible.
Since repeated requests do not alter the state of \LRU's cache,
$\LRU(I_C)=\LRU(I_C')$. However, then $I_C$ is a worse permutation
of $I$ than $I_{\LRU}$, which is a contradiction.

By the assumption in the statement of the lemma,
$\LRU(I')\leq\FIFO^G_W(I')$.
Let $I_{\FIFO}'$ be a worst ordering of $I'$ for \FIFO,
so $\FIFO(I_{\FIFO}')=\FIFO^G_W(I')$.
Again, we can insert pages removed
from $I_{\LRU}$ compared to $I'$ into $I_{\FIFO}'$, creating $I_{\FIFO}$,
i.e., inserting any removed request to $p$ immediately after an
existing request to $p$ in $I_{\FIFO}'$.
This will not change the state of the cache of \FIFO at any point in time,
so $\FIFO(I_{\FIFO})=\FIFO(I_{\FIFO}')$.
Thus,
\[\LRU^G_W(I)\leq\LRU(I')\leq\FIFO^G_W(I_{\FIFO}')=
 \FIFO(I_{\FIFO})\leq\FIFO^G_W(I)\]
\end{proof}

\begin{corollary}
Let $G$ be any access graph.
Assume that for all $I$, where there exists a worst ordering
$I_{\LRU}\in\WORST{I}{G}{\LRU}$
such that $I_{\LRU}$ has no two consecutive requests to the same page,
$\LRU(I_{\LRU})\leq\FIFO^G_W(I)$.
Then, for all $I$, $\LRU^G_W(I)\leq\FIFO^G_W(I)$.
\end{corollary}
\begin{proof}
This follows from the above by repeatedly removing the $j-1$ hits
in a sequence of $j$ consecutive requests to the same page. %
\end{proof}
We have now established the following property:

\begin{property}
\label{property-no-consecutive}
In proving for any access graph $G$, any sequence $I$ respecting $G$,
and any $I_{\LRU}\in\WORST{I}{G}{\LRU}$ that
$\LRU(I_{\LRU})\leq\FIFO^G_W(I)$,
we may assume that $I_{\LRU}$ has no consecutive
requests to the same page.
\end{property}

We now give a collection of definitions enabling us to be precise about
how a request sequence without consecutive requests to the same
page moves around on the cycle.

\begin{definition}
\mbox{}
\begin{itemize}
\item
An \textit{arc} is a connected component of a cycle graph.
As a mathematical object, an arc is the same as a path (in this section),
but refers to a portion of $C_N$, rather than a part of the walk
defined by a request sequence.
\item
One can fix an orientation in a cycle so that
the concepts of moving in a clockwise or anti-clockwise
direction are well-defined. We refer to a walk as being
{\em uni-directional} if each edge is traversed in the same direction
as the previous, and abbreviate this {\em u-walk}.
\item
A request $r_i$ in the request sequence is a {\em turn}
if the direction changes at that vertex,
i.e., if $r_i$ is neither the first nor the last request and $r_{i-1}=r_{i+1}$. The vertex requested is referred to as a {\em turning point}.
\item
When convenient we will represent a request sequence $I$ by
its {\em turn sequence},
\[T=\SEQ{A_1,v_1,A_2,v_2,\ldots,A_z,v_z},\]
where $T=I$,
$v_z$ is simply the last request of the sequence,
all the other $v_i$'s are the turns of the request sequence,
and all the $A_i$'s are u-walks.
Thus, for all $i<z$, either $A_i \subseteq A_{i+1}$ or $A_{i+1} \subseteq A_i$.
We refer to a turn $v_i$ as a {\em clockwise (anti-clockwise) turn}
if the $A_{i+1}$ goes in the clockwise (anti-clockwise) direction.
\item
Two turns are said to be {\em opposite} if they are in different directions.
\item
If for some $i<z$, $|A_{i+1}\cup\SET{v_{i+1}}| \geq k$,
then $v_i$ is an {\em extreme turn}.
Otherwise, $v_i$ is a {\em trivial turn}. 
\end{itemize}
\end{definition}

Most of the above is obvious terminology about directions around
the circle. The last definition, on the other hand, is motivated
by the behavior of the paging algorithms that we analyze.
Not surprisingly, it turns out to be an important distinction whether
or not the cache will start evicting pages before turning
back. We treat this formally below.

Our first aim is to ensure that all u-walks have length~$k+1$,
including the turning vertices. This is basically obtained by
removing all trivial turns. However, the first part is a special
case that we deal with first.

\begin{lemma}\label{lemma-first-walk-long}
Assume Property~\ref{property-no-consecutive}.
For the access graph $C_N$, assume that for any $I$ and
$I_{\LRU}\in\WORST{I}{C_N}{\LRU}$,
where $I_{\LRU}$ has turn sequence $\SEQ{A_1,v_1,A_2,v_2,\ldots,A_z,v_z}$
and $|A_1|\geq k-1$,
we have that $\LRU(I_{\LRU})\leq\FIFO^{C_N}_W(I)$.
Then, for any $I$,
$\LRU^{C_N}_W(I)\leq\FIFO^{C_N}_W(I)$.
\end{lemma}

\begin{proof}
Assume we are given $I$ and consider $I_{\LRU}\in\WORST{I}{C_N}{\LRU}$.
We may assume that $I_{\LRU}$ has no repeated requests to the same page.
If $|A_1|\geq k-1$, then we are done.
Otherwise, consider the turn sequence of
$I_{\LRU}$, $\SEQ{A_1,v_1,A_2,v_2,\ldots,A_z,v_z}$.

Let $w$ be the first fault for \LRU that occurs after $v_1$,
if any more faults occur.
The vertex $w$ could be a neighbor of the first vertex in $A_1$
or a neighbor of $v_1$.

If $w$ is a neighbor of the first vertex in $A_1$, we eliminate
$A_1$ from the sequence. The sequence still has the same number of
faults and the state of \LRU's cache at $w$ is unchanged,
so the result follows from Lemma~\ref{lemma-reduce}.

If $w$ is a neighbor of $v_1$, then we eliminate the subsequence
starting immediately after the first request to $v_1$ up until,
but not including, $w$.
Again, this sequence incurs the same number of faults as before
and leaves the cache state at $w$ as it was without this change,
so the result again follows from Lemma~\ref{lemma-reduce}.

Note that in the reduction just described, we are removing at least
one turn. Thus, we can repeat this process inductively until the
sequence leading to the first turn has the desired length.

Also note that we may end up in a trivial case, where we eliminate
all turns, and the remaining one u-walk has length less then $k$.
In that case, we are of course done with the entire proof of this
section, since all algorithms fault on all requests in such
a sequence.
\end{proof}

We have now established the following property:
\begin{property}
\label{property-long-first-walk}
We may assume that a worst ordering for \LRU is of the form
\[\SEQ{A_1,v_1,A_2,v_2,\ldots,A_z,v_z}, \textrm{ where } |A_1|\geq k-1. \]
\end{property}

We now reduce our problem to sequences without trivial turns.

\begin{lemma}\label{lemma-prop-no-trivial-turns}
Assume Property~\ref{property-no-consecutive}~and~\ref{property-long-first-walk}.
For the access graph $C_N$,
assume that for any $I$ and
$I_{\LRU}\in\WORST{I}{C_N}{\LRU}$,
where $I_{\LRU}$ has no trivial turns,
we have that $\LRU(I_{\LRU})\leq\FIFO^{C_N}_W(I)$.
Then, for any $I$,
$\LRU^{C_N}_W(I)\leq\FIFO^{C_N}_W(I)$.
\end{lemma}

\begin{proof}
Assume we are given $I$
and consider $I_{\LRU}\in\WORST{I}{C_N}{\LRU}$.
We may assume that $I_{\LRU}$ has no repeated requests to the same page.
If $I_{\LRU}$ has no trivial turns, then we are done.
Otherwise, consider the turn sequence of
$I_{\LRU}$, $\SEQ{A_1,v_1,A_2,v_2,\ldots,A_z,v_z}$,
and assume that $v_i$ is the first trivial turn.
Let $w$ be the first fault for \LRU that occurs after $v_{i+1}$,
if any more faults occur.

Assume that $v_i$ was entered from the direction $d$
(which is either clockwise or anti-clockwise).
\begin{description}
\item[$w$ is reached from direction $d$:]
Since $v_i$ is the first trivial turn
and since we know that $|A_1\cup\SET{v_1}|\geq k$,
we must have that $|A_i\cup\SET{v_i}|\geq k$.

Since $w$ is a fault, $w$ must be a neighbor of $v_i$ in direction $d$.
Thus, $I$ can be written
\[I = \SEQ{\ldots, v_{i+1}', B, v_i, A_i, v_{i+1}, B', v_i', w, \ldots}\]
where the unmarked $v_i$ and $v_{i+1}$ are turning points,
the dashed $v_i$ and $v_{i+1}$ are requests to the same vertices
as indicated by the index,
$B$ is a u-walk, and $B'$ is a walk (which could possibly contain turns).
We define $I'$ as
\[I' = \SEQ{\ldots, v_{i+1}', B', v_i', w, \ldots}\]
Thus, we have eliminated at least two turns,
and, in particular, at least one trivial turn.
We have only removed hits.
In addition, the cache content, including information concerning which
pages are least recently used, is exactly the same just before
$w$ in $I'$ as it was just before $w$ in $I$,
since all removed requests have been requested in
$\SEQ{v_{i+1}', B', v_i'}$.
In fact, $v_i$ is the most recently used,
and, following the arc in the opposite direction of $d$,
pages are less and less recently used.
By Lemma~\ref{lemma-reduce}, we have reduced the problem to considering
$I'$ instead of $I$.
\item[$w$ is reached from the direction opposite $d$:]
No request can have been made to the neighbor of $v_i$ in the direction
$d$, since then we would be in the case above.
Thus, $I$ must be of the form
\[I = \SEQ{\ldots, v_i, A_i, v_{i+1}, B', w, \ldots}\]
where $B'$ is a walk that contains an odd number of turns.
We define $I'$ as
\[I' = \SEQ{\ldots, v_i, B, w, \ldots}\]
where $B$ is the arc such that
$\SET{B}=\SET{A_i}\cup\SET{v_{i+1}}\cup\SET{B'}$.
Thus, we have eliminated at least two turns,
and, in particular, at least one trivial turn (at least two, actually).
We have only removed hits.
In addition, the cache content, including information concerning which
pages are least recently used, is exactly the same just before
$w$ in $I'$ as it was just before $w$ in $I$,
since all removed requests have been requested in
$\SEQ{v_i', B}$.
In fact, $v_i$ is the least recently used,
and, following the arc in the opposite direction of $d$,
pages are more and more recently used.
By Lemma~\ref{lemma-reduce}, we have reduced the problem to considering
$I'$ instead of $I$.
\end{description}
In either case, we have reduced the problem to one with fewer trivial turns.

We now consider the remaining case where
there were no more faults (such that no such $w$ exists).
In that case, $I'$ is simply the sequence $I$ cut off after
the trivial turn $v_i$, and everything holds similarly.

By induction, we can clearly apply this method repeatedly until all
trivial turns have been removed.
\end{proof}

We have now established the following property:

\begin{property}
\label{property-long-paths}
We may assume that a worst ordering for \LRU is of the form
\[\SEQ{A_1,v_1,A_2,v_2,\ldots,A_z,v_z}, \textrm{ where } \forall i\WEHAVE |A_i|\geq k-1. \]
\end{property}

If these properties hold for some sequence, $I$, then it is easy to 
see that the number of turns determines how many hits \LRU has on $I$.

\begin{proposition}
\label{proposition-extreme-hits}
If $I$ has the form of Property~\ref{property-long-paths} and contains no repeated requests to the same page 
then \LRU has exactly $(z-1)(k-1)$ hits on $I$.
\end{proposition}

Next we show that we may assume that in a worst ordering for
\LRU, there is no turn which is followed by going
all the way around the cycle in the opposite direction.

\begin{definition}
Let $u$, $v$, and $w$ be three distinct consecutive vertices on $C_N$.
We refer to $I$ as having an {\em overlap} if $I$ can be written
$\SEQ{\ldots u,v,u,B,w,v \ldots}$. If $I$ does not have an overlap,
we refer to $I$ as {\em overlap-free}.
\end{definition}

\begin{lemma}
Assume Properties~\ref{property-no-consecutive}--\ref{property-long-paths}.
For the access graph $C_N$,
assume that for any $I$ and $I_{\LRU}\in\WORST{I}{C_N}{\LRU}$,
where $I_{\LRU}$ is overlap-free,
we have that $\LRU(I_{\LRU})\leq\FIFO^{C_N}_W(I)$.
Then, for any $I$, $\LRU^{C_N}_W(I)\leq\FIFO^{C_N}_W(I)$.
\end{lemma}
\begin{proof}
Let $I_{\LRU}\in\WORST{I}{C_N}{\LRU}$.
If $I_{\LRU}$ has an overlap, we show that by reordering while
respecting $C_N$ an overlap-free sequence with at least as many faults
can be constructed.

Assume that $I_{\LRU}$ has an overlap
and consider a first occurrence of a vertex 
$u$ in $I_{\LRU}$ such that $I_{\LRU}$
contains the pattern $\SEQ{\ldots,u,v^1,u,B,w,v^2,\ldots}$,
where $u$, $v$, and $w$ are consecutive vertices on $C_N$.
The superscripts on $v$ are just for reference,
i.e., $v^1$ and $v^2$ are the same vertex.

We define $I'=\SEQ{\ldots,u,v^1,w,B^R,u,v^2,\ldots}$,
where $B^R$ denotes the walk $B$, reversed. Clearly, $I'$ respects $C_N$.
We now argue that $I'$ incurs no more faults than $I_{\LRU}$.
Clearly, there is a turn at $v^1$ in $I_{\LRU}$.
If there is also a turn at $v^2$, then we have effectively just removed
two turns. According to Proposition~\ref{proposition-extreme-hits},
$I_{\LRU}$ cannot be a worst ordering then.
Thus, we can assume there is no turn at $v^2$.

In the transformation, we are removing the turn at $v^1$
and introducing one at $v^2$.
Thus, since in the sequence $I_{\LRU}$ all u-walks between turns contained
at least $k-1$ vertices, this is still the
case after the transformation in $I'$,
except possibly for the u-walk from the newly created turn at $v^2$ to the next turn in the sequence. Let $x$ denote 
such a next turn.

If the u-walk between $v$ and $x$ has at least $k-1$ vertices, then
the transformed sequence has the same number of turns, all u-walks
between turns contain at least $k-1$ vertices,
and therefore $I_{\LRU}$ and $I'$ have the same number of hits (and faults).
In addition, the state of the caches after treating $I_{\LRU}$ up to $x$
and $I'$ up to $x$ are the same.

If that u-walk contains fewer than $k-1$ vertices,
we consider the next turn $y$ after $x$.
Since there are at least $k-1$ vertices
in between $x$ and $y$, we must pass $v$ on the way to $y$.

Thus, we are now considering \[I_{\LRU=}\SEQ{\ldots,u,v^1,u,B,w,v^2,B_1,x,B_2,v^3,B_3,y,\ldots}\]
where there are turns at $v^1$, $x$, and $y$,
versus \[I'=\SEQ{\ldots,u,v^1,w,B^R,u,v^2,B_1,x,B_2,v^3,B_3,y,\ldots}\]
where there are turns at $v^2$, $x$, and $y$.

Comparing $\SEQ{\ldots u,v^1,u,B,w,v^2}$ with $\SEQ{\ldots u,v^1,w,B^R,u,v^2}$,
one observes that both sequences have least $k-1$ vertices on any u-walk between
two turns, and
the latter has one fewer turns. Thus, by
Proposition~\ref{proposition-extreme-hits}, it has $k-1$ fewer hits.

By assumption, $B_1$ has fewer than $k-1$ vertices. Thus, comparing
$I_{\LRU}$ and $I'$ up to and including $x$, $I'$ has at least as many
faults.

In $I_{\LRU}$, $\SEQ{B_2,v^3}$ must all be hits, so up to and including
$v^3$, $I'$ has at least as many faults.

Since the u-walk leading to $v^1$ in $I'$ contains at least $k-1$ vertices
(not including $v^1$), and since the u-walk going from $v^2$ to $y$ goes
in the same direction, the requests in $\SEQ{B_3,y}$ must all be faults
in $I'$.

Thus, we have shown that there are at least as many faults in $I'$
as in $I_{\LRU}$.
In addition, the state of the caches after treating $I_{\LRU}$ up to $y$
and $I'$ up to $y$ are the same.

With the transformation above, we do not incur more faults, and
any first occurrence of a vertex
$u$ initiating an overlap pattern has been moved further towards
the end of the sequence. Thus, we can apply this transformation
technique repeatedly until no more such patterns exist.  %
\end{proof}
We have now established the following property:

\begin{property}
\label{property-overlap-free}
We may assume that a worst ordering is overlap-free.
\end{property}

Now we have all the necessary tools to prove the theorem of this section.

\begin{theorem}
\label{cycles-leq}
For all $I$ respecting the access graph $C_N$,
\[\LRU^{C_N}_W(I) \leq \FIFO^{C_N}_W(I). \]
\end{theorem}

\begin{proof}
We may assume
Properties~\ref{property-no-consecutive}--\ref{property-overlap-free}.

Consider any $I$ and $I_{\LRU}\in\WORST{I}{C_N}{\LRU}$.
If there are no turns at all in $I_{\LRU}$, both \FIFO and \LRU will fault
on every request.
If there is only one turn,
\FIFO will clearly fault as often as \LRU on $I_{\LRU}$,
since we may assume that there is no overlap.

So, consider the first two turns $v$ and $v'$.
By Property~\ref{property-overlap-free}, we cannot have the pattern
$\SEQ{\ldots,u,v,u,B,w,v,\ldots}$.
Thus, after the first turn, the edge from $w$ to $v$
can never be followed again.
This holds symmetrically for $v'$, which is a turn in the other direction.
Thus, once the request sequence enters the arc between $v$ and $v'$,
it can never leave it again.
We refer to this arc as the {\em gap}.
To be precise, since we are on a cycle, the gap is the arc that
at the two ends has the neighbor vertices of $v$ and $v'$
from which edges to $v$ and $v'$, respectively, cannot be followed again,
and such that $v$ and $v'$ are not part of the arc.

Assume without loss of generality that,
after the first turn,
if the request sequence
enters the gap between $v$ and $v'$, then it does so coming from $v'$.
Thus, after the first turn at $v$, the requests can be assumed to
be given on the path access graph $P_N$ instead of the cycle $C_N$,
where the access graph
$P_N$ starts with $v$ and continues in the direction of the turn at $v$
and ends at the neighbor of $v$ in the gap.

In fact, we can assume that we are working on the access graph $P_N$
from $k-1$ requests before the first turn at $v$,
since all u-walks can be assumed to have at least that length.
Let $r_i$ be that request.
Since there are no turns before $v$, starting with $r_i$,
\LRU and \FIFO function
exactly as they would starting with an empty cache.

We divide
$I_{\LRU}=\SEQ{r_1, r_2, \ldots,  r_{|I_{\LRU}|}}$ up into
the sequences
$\SEQ{r_1, r_2, \ldots, r_{i-1}}$ and $\SEQ{r_i, \ldots, r_{|I_{\LRU}|}}$.
Here, the former is a u-walk, where \LRU and \FIFO both fault on every
request, and the latter can be considered a request sequence on a
path access graph as explained above, and the
conclusion follows from Theorem~\ref{paths-leq}.  %
\end{proof}

\section{Separation on a path of length $\mathbf{\emph{k}+1}$ }
In the last sections, we showed that \LRU was at least as good as \FIFO
on any path graph or cycle graph. Now we show that \LRU is strictly better
if these graphs contain paths of length at least $k+1$. 
We exhibit a family of sequences
$\SET{I_n}_{n \geq 1}$ such that 
$\FIFO_W^{P_N}(I_n) \geq \big(\frac{k+1}{2}\big)\cdot\LRU_W^{P_N}(I_n) + b$,
for some fixed constant $b$, on
path graphs $P_N$ with $N \geq k+1$. Only $k+1$ different pages are
requested in $I_n$.
The same family of sequences is also used to show that \FWF is worse
than either \LRU or \FIFO.
We number the
vertices of the path graph $P_N$ in order from $1$ through $N$.

In order to get an exact value for the number of faults \FIFO has
on its worst ordering of $I_n$, we first prove an upper bound which holds for these reorderings.

\begin{lemma}
\label{fault_2k}
On any sequence respecting the path
graph, $P_{k+1}$, \FIFO incurs at most $k+1$ faults
on any $2k$ consecutive requests.
\end{lemma}

\begin{proof}
Since \FIFO is conservative, a subsequence consisting of $k$ distinct
pages can give rise to at most $k$ faults.
Hence, for at least $k+1$ faults to occur,
the sequence must visit both endpoints of the path graph.

The $(k+1)$st fault leads to the eviction of the page $p$ requested
at the first fault.
We now argue that if a $(k+2)$nd fault occurs, then the subsequence
of consecutive requests has length at least $2k+1$.

Since the size of the graph is $k+1$, the request $r$ giving
rise to a $(k+2)$nd fault,
must be on the next request for $p$.
Therefore, if $p$ is an endpoint, then 
the request sequence consists of a walk to the other endpoint and back again.
If $p$ is not an endpoint,
then the request sequence must be a walk in which the two faults on requests
for $p$ are separated by requests to each of the endpoints.
In either case, the walk must be of length at least $2k+1$.
\end{proof}

We use the above lemma to analyze a family of sequences
and the performance of \FIFO and \LRU on any reordering
respecting the access graph $P_{k+1}$.
The same sequence family will 
also yield separation results between \LRU and \FWF,
as well as between \FIFO and \FWF. 

We define $I_n = \SEQ{1, 2, \ldots, k, k+1, k, k-1, \ldots, 2}^n$. Each block
$\SEQ{1, 2, \ldots, k, k+1, k, k-1, \ldots, 2}$ in $I_n$ contains
$2k$ page requests.

\begin{itemize}
\item
$I'\in\WORST{I'}{G}{\LRU}$, and
\item
if $\LRU(I')\leq\FIFO^G_W(I')$, then
$\LRU^G_W(I)\leq\FIFO^G_W(I)$.
\end{itemize}
The following result, is similar to a result shown in~\cite{BIRS95},
comparing the behavior of \FIFO to \LRU. 

\begin{lemma}
\label{fifo_thm}
Let $I_n = \SEQ{1, 2, \ldots, k, k+1, k, k-1, \ldots, 2}^n$.
Then \[\FIFO_W^{P_{k+1}}(I_n) = (k+1)n.\]
\end{lemma}

\begin{proof}
We begin by showing that $\FIFO(I_n) =(k+1)n$.
We denote the prefix of each block,
$\SET{1,2,\ldots, k,k+1}$ by $I_1$ and
the suffix $\SET{k,k-1,\ldots, 2}$ by $I_2$,
and define block a block $B = \SEQ{I_1, I_2}$.
So, $I_n = \SEQ{B}^n$. We analyze  
the first block $B_1$ and show that subsequent blocks generate exactly
the same faults.

In $B_1$, while processing $I_1$,
there are $k+1$ faults and the resulting cache configuration is
$(2,3,\ldots,k+1)$, where page $i$ is brought into cache before $j$ for 
all $j > i$ and the only page outside the cache is $1$.
As a result, \FIFO does not fault while processing $I_2$.
All through $I_1$ in the next block, $B_2$, \FIFO incurs 
only faults, ending with the eviction of $1$ at the request to $k+1$.
Note that the cache configuration is the same as the one at the end
of $I_1$ in $B_1$.
Repeating this, the cache configuration is the same after the treatment
of each block, and the total number of faults is $(k+1)n$.

By Lemma~\ref{fault_2k}, \FIFO cannot incur more than $(k+1)n$ faults on any
sequence of length $2kn$ respecting $P_{k+1}$, so the result follows.
\end{proof}

We now consider \LRU's performance on its worst reordering of $I_n$.

\begin{lemma}
\label{lru_thm}
If $N \geq k+1$, then for the sequence
$I_n = \SEQ{1,2,\ldots,k,k+1,k,k-1,\ldots,2}^n$, we have $\LRU^{P_N}_W(I_n) =\LFD^{P_N}_W(I_n) = 2(n-1) + k+1$.
\end{lemma} 

\begin{proof}
The first $k +1$ faults are due to the initial requests
when the cache is not full.
Any reordering of $I_n$ respecting the access graph
will involve $2n$ requests to 
each page in $\SET{2, 3, \ldots, k-1, k}$
and $n$ requests to $1$ and $k+1$.
Any reordered sequence must also respect the path access graph $P_N$
and any walk between $1$ and $k+1$ must pass through $k-1$ other vertices.
If there is a fault on $1$ or $k+1$, respectively,
then the cache must contain the other $k$ pages and \LFD will 
evict $k+1$ or $1$, respectively,
and not incur any faults on the intermediate requests.
Therefore, overall \LFD incurs a total of $2(n-1) + k+1$ faults on any 
reordered sequence, and thus on the worst reordering as well.

Since, by Theorem~\ref{lru-opt},
\LRU's performance equals that of \LFD's on a path access graph,
the result follows.
\end{proof}

The difference between \FIFO's and \LRU's performance on $I_n$ gives
the desired separation.

\begin{theorem}
\label{paths-better}
For $N\geq k+1$, there exists a family of sequences $\SET{I_n}$ respecting the access graph
$P_N$ and a constant $b$ such that the following two conditions hold:
\[
\lim_{n\rightarrow\infty}\LRU(I_n)=\infty
\mbox{ and for all $I_n$, }
\FIFO^{P_N}_W(I_n)\geq\big(\frac{k+1}{2}\big)\cdot \LRU_W^{P_N}(I_n) +b.
\]
\end{theorem}

\begin{proof}
Follows from Lemmas~\ref{fifo_thm} and~\ref{lru_thm} with $b = 1-k$.  %
\end{proof}

Next, we prove a tight upper bound on the relative worst order ratio
of \FIFO and \LRU
for path access graphs. Note that there exist sequences respecting
the line, where \LRU does not  fault
at least twice whenever \FIFO faults $k+1$ times.
Let $I_s = \SEQ{S_0,S_1,...,S_s}$ where
$S_i = \SEQ{i+k,i+k-1,...,i+2,i+1,i+2,...,i+k-1,i+k}$.
\LRU faults on the first $k$ pages and then the first
page in every $S_i$ after that. \FIFO faults on the
first $k$ requests in every $S_i$. So \LRU faults $k+s$ times
and \FIFO faults $k+ks$ times. However, there are always reorderings
of the sequence where \LRU does fault this much.

\begin{lemma}
\label{lru-P}
For $N\geq 1$ and any sequence $I$ respecting $P_N$, we have that
\[\FIFO^{P_N}_W(I) \leq \left(\frac{k+1}{2}\right)\cdot\LRU^{P_N}_W(I).\]
\end{lemma}

\begin{proof}
The result is trivial if $k=1$ or $N\leq k$, so assume that $k\geq 2$
and $N\geq k+1$.

Consider any sequence $I$ respecting the path, $P_N$.
We divide $I$, except for a possible suffix,
up into a number of blocks, $B_1,B_2,...,B_m$.
The first block, $B_1$, starts with the first request of $I$ continuing
up to and including the request where \FIFO would fault for
the $(k+1)$st time.
Block, $B_i$ for $i\geq 2$ starts with the first
request not included in the previous block, $B_{i-1}$, and continues
up to, and including, the request where \FIFO would fault for
the $(k+1)$st time in $B_i$.

Note that since the sequence considered respects the path $P_N$,
 any block, $B$, of consecutive requests defines
an interval of the line $P_N$ in a natural way. 
The interval consists of all of the pages requested
in the block, and there are no holes in the interval because
the sequence respects the path. The {\em endpoints} of the block
are the pages which are the endpoints of the interval.

This definition of blocks may leave a remainder of requests
in $I$ not included in a block.
We deal with that at the end of the proof. Temporarily remove
these last requests from $I$ and call the resulting sequence $I'$.
\FIFO faults $m(k+1)$ times on $I'$.

We show how reorder $I'$, block by block, creating a sequence, $J$,
which is partitioned into the same number
of blocks, $T_1,T_2,...,T_m$, so that \LRU faults
at least two times in each of these $m$ blocks.
Thus, \LRU will fault
$2m$ times on this reordering of $I'$, giving the desired result
asymptotically.

In some cases, $B_i$ and $T_i$ will be identical. When not,
they will end with the same request and the rest of
the block will be in the reverse order. In this latter case,
if $B_i=\SEQ{r_{i_1},r_{i_2},...,r_{i_{q-2}},r_{i_{q-1}},r_{i_{q}}}$,
then $T_i=\SEQ{r_{i_{q-1}},r_{i_{q-2}},...,r_{i_2},r_{i_1},r_{i_{q}}}$,
which we denote by $B^R_i$.
We show later that this is well-defined, i.e., that it leads to a sequence
respecting the access graph.

Let $T_1=B_1$.
\LRU faults $k+1>2$ times on $T_1$.

Consider any block, $B_i$, $i\geq 2$, in $I'$. 
We use the fact that \FIFO is conservative~\cite{BE97b}.
By definition, this means that on any subsequence with $k$
pages, it makes at most $k$ faults.
Thus, given that it faults $k+1$ times
in each block, there must be at least $k+1$
distinct pages in each block. 

Consider running \LRU
on the sequence defined by $\SEQ{T_1,T_2,...,T_{i-1},B_i}$. 
If \LRU faults at least twice in $B_i$, then let $T_i=B_i$. 
If there are $k+2$ pages in
$B_i$, \LRU must fault at least twice, since it only
has $k$ pages in cache at the start of the block. 
Now, assume that \LRU faults at most once in $B_i$ and thus that $B_i$
only has $k+1$ distinct pages.

In this case, we let $T_i=B^R_i$.
Consider the last page, $q$, requested in $B_{i-1}$, which is also the 
last in $T_{i-1}$.
If $q$ is not in $B_i$, then, by assumption,
there are exactly $k+1$ pages from and not including $q$
to and including the furthest point $z$ in $B_i$.
Since \LRU has $q$ in cache immediately before treating $B_i$,
it has at most $k-1$ of the $k+1$ pages from $B_i$ in cache,
and must fault twice on $B_i$, contradicting our assumption.
Thus, $q$ must be in $B_i$.

Since $q$ is in $B_i$ and \FIFO faults on every page in the
interval defined by $B_i$, \FIFO faults on this request to $q$ in $B_i$.
To do this, it must have faulted
on $k$ different pages since the fault on $q$ last in $B_{i-1}$,
so, by the definition of blocks, $q$ must be the last page in $B_i$, too.
This establishes that if
$\SEQ{T_1,T_2,...,T_{i-1},B_i}$ respects the access graph,
then $\SEQ{T_1,T_2,...,T_{i-1},T_i}$ does too.

Now consider how many times \LRU faults on $B^R_i$.
The block $B_i$ has two endpoints, $s$ and $t$,
with $k-1$ distinct pages between them. 
Without loss of generality, assume that $B_i$ has the form
$\SEQ{r_1,r_2,...,r_i,s,r_{i+2},...,r_j,t,r_{j+2},...,r_{\ell},q}$, where
the occurrences of $s$ and $t$ are the first such. By assumption, \LRU
faults at most once on $B_i$, so it does not fault on both $s$ and $t$.
Given the number of pages between $s$ and $t$, by definition of \LRU,
after a request to one of these pages, it must
fault on the next request to the other page.
Thus, in order for \LRU to fault at most once on $B_i$, there must have been
a request to $s$ in $T_{i-1}$ and there cannot have been a request
to $t$ in $T_{i-1}$ after the last request to $s$.
In $B_i$, there cannot be a request to $s$ after the request to $t$,
since then \LRU faults twice, contrary to our assumption.
Thus, in $T_i=B^R_i$, there will be a request to $t$ before the request to
$s$, so \LRU will fault on both of these.

Having established that the asymptotic ratio is two, we return
to the possible suffix of $I$ after the last block, call it $I''$.
 \FIFO faults at most $k$ times on $I''$ or it would be a complete
block.
First, if $k\geq 3$, then \LRU faults at least four times on the
first block. Thus, there are two extra faults which will bring
\LRU's total up to enough to cover the possible lack of faults on $I''$.
Only the case $k=2$ remains.
In this case, there is only one extra fault for \LRU in the first block which
can be used to cover the faults required for $I''$.
If \FIFO faults only once in that last part, the ratio will still
be less than $\frac{k+1}{2}$.
Suppose \FIFO faults $k=2$ times. It faulted on the last page
in $B_m$, which must be different from these two pages in $I''$.
That last page in $B_m$ is also the last page in $T_m$, so \LRU must
have it in cache at the start of $I''$. Thus, it must fault on at
least one of the two pages \FIFO faults on there, giving the extra
fault necessary to avoid an additive constant.
\end{proof}

We now have tight upper and lower bounds on the relative worst
order ratio of \FIFO to \LRU on paths.

\begin{theorem}\label{fifo-lru-thm}
If $N \geq k+1$, then the relative worst order ratio of
\FIFO to \LRU on the path access graph is
$\RWOR{\FIFO}{\LRU}^{P_N} = \frac{k+1}{2}$.
\end{theorem}

\begin{proof}
Referring to the definition of relative worst order ratio from
Section~\ref{section-preliminaries},
Theorem~\ref{paths-leq} shows that $c_l(\FIFO, \LRU) \geq 1$.
Therefore, $\RWOR{\FIFO}{\LRU}=c_u(\FIFO, \LRU)$. Theorem~\ref{paths-better} 
implies that $\RWOR{\FIFO}{\LRU}^{P_N} \geq \big(\frac{k+1}{2}\big)$
and Lemma~\ref{lru-P} gives the equality.
\end{proof}

The following lemma and its corollary, showing that \FWF is never
better than \FIFO or \LRU, are quite possibly folklore:
\begin{lemma}
\label{fwf-fifo}
For any sequence $I$ and any conservative algorithm
\ALG{A}, we have $\ALG{A}(I) \leq \FWF(I)$.
\end{lemma}

\begin{proof}
Given a sequence $I$, divide it up into $k$-phases as described
in~\cite{BE97b}:
Phase~$0$ is the empty sequence. For every $ i \geq 1$, Phase~$i$ is a 
maximal sequence following phase $i-1$ that contains at most $k$
distinct page requests. Phase~$i$ begins on the $(k+1)$st distinct
page requested since the start of Phase~$i-1$.

It is easy to see that \FWF flushes at the first request of every
Phase~$i$, $i > 1$, and hence incurs $k$ faults on the set of distinct
requests within each phase. By definition, no
conservative algorithm can fault more than $k$ times in any $k$-phase.
\end{proof}

\begin{corollary}
\label{fwf-fifo-worst}
$\FIFO_W(I) \leq \FWF_W(I)$ and $\LRU_W(I) \leq \FWF_W(I)$.
\end{corollary}
\begin{proof}
Follows directly since \LRU and \FIFO are conservative algorithms. %
\end{proof}

The separation showing that \FWF is strictly worse than these conservative
algorithms on any graph containing $P_{k+1}$ uses the family of
sequences $I_n$.
\begin{lemma}
\label{fwf-upper}
\FWF incurs a fault on every request in
\[I_n=\SEQ{1,2\ldots, k,k+1, k, k-1, \ldots,3, 2}^n,
\textrm{ giving } \FWF(I_n) = 2kn. \] 
\end{lemma}

\begin{proof}
A flush occurs at $k+1$ in the first encounter of that page,
and then at $1$ at the beginning of the next repetition.
The same process repeats itself in every repetition,
flushing at $1$ and $k+1$.
Hence, \FWF faults on every request and $\FWF(I_n) =2kn$. %
\end{proof}

 It was shown in \cite{BFL07j} that for a complete graph, 
the relative worst order ratio of \FWF to \FIFO
is exactly $\frac{2k}{k+1}$. This is also a lower
bound for any graph containing $P_{k+1}$, but it is still open to
determine if equality occurs in all sparser graphs or not.

\begin{theorem}\label{fwf-fifo-thm}
For any access graph $G$ which has a path of length at least $k+1$,
\[\RWOR{\FWF}{\FIFO}^G \geq \frac{2k}{k+1}.\]
\end{theorem}
\begin{proof}
Follows from Lemma~\ref{fifo_thm}, Corollary~\ref{fwf-fifo-worst}, and Lemma~\ref{fwf-upper}.
\end{proof}

The relative worst order ratio of \FWF to \LRU on paths is exactly $k$.

\begin{theorem}\label{fwf-lru-thm}
For any access graph $G$ which has a path of length at least $k+1$,$\RWOR{\FWF}{\LRU}^G = k$.
\end{theorem}

\begin{proof}
By Corollary~\ref{fwf-fifo-worst}, for any sequence $I$,
$\LRU_W(I) \leq \FWF_W(I)$.

We now argue that for any request sequence $I$, $\FWF(I) \leq k \cdot \LRU(I)$.
We decompose the sequence into $k$-phases as described in the
proof of Lemma~\ref{fwf-fifo}. As argued there, \FWF will flush at the
beginning of every phase and therefore must 
incur $k$ faults in each phase. \LRU faults on the first request
of each phase since the $k$ distinct pages from the previous phase
have been requested more recently.
Thus, if \FWF incurs $kx$ faults, then \LRU will incur at least $x$,
and so $\FWF(I) \leq k \cdot \LRU(I)$,
implying that $\RWOR{\FWF}{\LRU} \leq k$.

From the sequence family $\SET{I_n}$ and
Lemmas~\ref{lru_thm} and~\ref{fwf-upper},
we obtain that $\RWOR{\FWF}{\LRU}^G \geq k$. Hence,
$\RWOR{\FWF}{\LRU}^G = k$.  %
\end{proof}

\section{Incomparability}

In this section, we show that on some general classes of access graphs,
\LRU and \FIFO are incomparable.

We consider the cyclic access graph defined by the edge set
\[\SET{ (1,2), (2,3), (3,4), (4,5), (5,1) }\]
using  a cache of size $4$ to process the
request sequence 
\[I_1 =  \SEQ{1,5,1,2,3,4,5,1,2,1}.\]

\begin{lemma}
\label{not-worst}
For cache size 4, on any reordering of $I_1$ respecting the cycle and starting with $1$,
\LRU incurs at least $8$ faults and
\FIFO incurs at most $7$ faults.
\end{lemma}

\begin{proof}
It is trivial to check that \LRU incurs $8$ faults
on the ordering given by $I_1$.

For \FIFO,
it is easy to check in the following that reorderings with repeated
requests do not lead to more faults by \FIFO.
The reorderings of $I_1$ either have a prefix of the type
$\BIGSETOF{ \SEQ{1,i,1} }{ i \in \SET{2,5} }$ or
$\BIGSETOF{ \SEQ{1, i, j} }{ i \ne j \ne 1 }$. 
For the latter, examples being $\SEQ{1, 2, 3}$ and $\SEQ{1, 5, 4}$,
the subsequence following the prefix contains $4$ distinct pages.
Since \FIFO is conservative, it can incur at most $4$ faults on that
part after the prefix, bringing the total fault count up to at most $7$.

The first four distinct page requests will always incur $4$ faults,
but for reorderings with the prefix
$\BIGSETOF{ \SEQ{1,i,1} }{ i \in \SET{2,5} }$,
some pages are repeated within the first four requests.
If the extended prefix is $\SEQ{1,i,1,i}$ for $ i \in \SET{2, 5}$,
then the rest of the sequence still contains $4$ distinct pages
and again can add at most $4$ faults to the previous $2$,
bringing the total up to at most $6$.
The only remaining case is a prefix of the form
$\SEQ{ 1,i, 1, j}$ where $i, j  \in \SET{2,5}$, $i \ne j$.
Here, there are $3$ faults on the prefix.
We divide the analysis of the rest of the sequence up into two cases
depending on the next request following $j$:

For the first case, if the next request is $1$,
the extended prefix is $\SEQ{ 1,i,1, j,1}$. 
However, then the next request to a page other than $1$
is either to $i$ or $j$ and therefore not a fault.
In addition, either there are no more $i$'s or no more $j$'s in the
remaining part of the sequence, and again \FIFO can then fault at
most $4$ times on this sequence with only $4$ distinct pages.

For the second case,
if the next request is $k \in \SET{3,4}$,
then visiting $l \in \SETOF{4,5}{l \ne k}$ before the next $j$
will give a prefix $\SEQ{ 1, i, 1, j, k , l }$ with 
$5$ faults, and the suffix must be $\SEQ{ i, 1, j, 1}$ or
$\SEQ{ i, 1 , 1, j }$, adding only one more fault. This gives
$6$ faults in total.
If $j$ is requested before $l$, the only possibilities are
$\SEQ{ 1,2,1,5,4,5,1,1,2,3 }$ and $\SEQ{ 1,5,1,2,3,2,1,1,5,4 }$.
In total, this gives only $5$ faults.  %
\end{proof}

Note that the result above does not contradict our result about cycles.
As predicted by that result, one of the worst orderings for
\LRU and \FIFO would be $\SEQ{ 2,1,5,4,3,2,1,5,1,1 }$,
incurring 8 faults for both algorithms.

Using the cycle graph on which we processed $I_1$,
we now construct a larger graph using ``copies'' of this graph as follows.
For $2 \leq i \leq n$, we define $I_i$ as a structural copy of $I_1$,
i.e, we use new page names, but with the same relative order as in $I_1$
(like putting a ``dash'' on all pages in $I_1$).
All these copies have their own set of pages
such that no request in $I_i$ appears in $I_j$ for $i \ne j$.
Just as $I_1$ implies a cycle graph that we denote $X_1$,
so do each of these sequences and we let
$X_i$ denote the graph implied by $I_i$.
Let $X_{i,k}$ denote the $k$th vertex in the $i$th copy
and $I_{j,k}$ denote the $k$th request in the $j$th copy.
To be precise, we define
$I_i =  \SEQ{ X_{i,1} , X_{i, 5}, X_{i, 1}, X_{i, 2} ,X_{i,3},X_{i,4}, 
                X_{i,5}, X_{i,1} ,X_{i,2},X_{i,1} }$.

We define a graph $\G_n$ with a vertex set containing all $X_{i,k}$ and
$n$ additional vertices $u_1, u_2, \ldots, u_n$.
Its edges are all the edges from the graphs $X_i$, $1\leq i\leq n$,
together with edges $(X_{i,1}, u_i)$ and $(u_i,X_{i+1,1})$
for all $i$, $1 \leq i \leq n-1$, plus the edge $(X_{n,1},u_n)$.

Thus, $\G_n$ can be described as a chain of cycles, where each two
neighboring cycles are separated by a single vertex. Clearly, the sequence $\I_n =  \SEQ{ I_1, u_1, I_2, u_2, I_3, u_3, \ldots, I_n, u_n } $
respects the access graph $\G_n$.

\begin{theorem}
\LRU and \FIFO are incomparable on the 
family of graphs $\SET{\G_n}$, 
according to relative worst order analysis. 
\end{theorem}

\begin{proof}
We use cache size $k=4$.
For the infinite family of sequences $\SET{\I_n}$
respecting the access graph $\G_n$,
the following two conditions hold:
\begin{itemize}
\item
  $\lim_{n\rightarrow\infty}\FIFO(\I_n)=\infty$.
\item
for all $\I_n$,
$\LRU^{\G_n}_W(\I_n) \geq \frac{9}{8} \cdot \FIFO_W^{\G_n}(\I_n)$.
\end{itemize}
The first condition obviously holds since entirely new pages are requested
as the sequences get longer.
With regards to the second condition,
since the requests $u_i$, $i \geq 1$, appear only once,
any permutation respecting $\G_n$ must have the following structure
or its reverse:
\[ \I_n'=  \SEQ{ I_1', u_1, I_2', u_2, I_3', u_3, \ldots, I_n', u_n } \]
where the sequence $I'_i$ is a reordering of $I_i$.
Note that for the reordering to respect the access graph,
each permutation $I_i'$ must begin and end with $X_{i,1}$.

By Lemma~\ref{not-worst},
none of the reorderings that start and end with $1$ give rise to more
than $7$ faults for \FIFO, while there is a reordering (the one given)
on which \LRU incurs $8$ faults.
Taking the vertices $u_i$ into account as well,
\FIFO incurs at most $8n$ faults and \LRU at least $9n$ faults
on any permutation of $\I_n$.
This proves that \LRU and \FIFO cannot be comparable in \LRU's favor.

On the other hand, consider the family of sequences 
\[ J_r = \SEQ{ X_{1,4}, X_{1,3}, X_{1,2}, X_{1,1}, u_1, X_{1,1},X_{1,2},X_{1,3},}^r. \]
The sequence constitutes a path on parts of $\G_n$.
There are $k+1$ pages in each repetition, so \LRU must fault
at least once per repetition.
Thus,
\[\lim_{r\rightarrow\infty}\LRU(J_r)=\infty.\]
By Theorem~\ref{paths-better} for $k=4$,
\[\FIFO^{\G_n}_W(J_r)\geq \left(\frac{k+1}{2}\right) \LRU^{\G_n}_W(J_r)-(k-1)
= \frac{5}{2}\LRU^{\G_n}_W(J_r)-3.\]
Thus, \LRU and \FIFO cannot be comparable in \FIFO's favor.

In conclusion, \LRU and \FIFO are incomparable.
\end{proof}

\begin{figure}
\begin{center}
\begin{picture}(308,64)
\put(28.531695,0.000000){\circle*{5}}
\path(28.531695,0.000000)(0.000000,20.729490)
\put(0.000000,20.729490){\circle*{5}}
\path(0.000000,20.729490)(10.898138,54.270510)
\put(10.898138,54.270510){\circle*{5}}
\path(10.898138,54.270510)(46.165253,54.270510)
\put(46.165253,54.270510){\circle*{5}}
\path(46.165253,54.270510)(57.063391,20.729490)
\put(57.063391,20.729490){\circle*{5}}
\path(57.063391,20.729490)(28.531695,0.000000)
\put(28.531695,0.000000){\circle*{5}}
\put(28.531695,11.000000){\makebox(0,0){\tiny $X_{1,1}$}}
\put(10.461622,24.128677){\makebox(0,0){\tiny $X_{1,2}$}}
\put(17.363776,45.371323){\makebox(0,0){\tiny $X_{1,3}$}}
\put(39.699615,45.371323){\makebox(0,0){\tiny $X_{1,4}$}}
\put(46.601769,24.128677){\makebox(0,0){\tiny $X_{1,5}$}}
\path(28.531695,0.000000)(68.531695,0.000000)
\put(68.531695,0.000000){\circle*{5}}
\put(68.531695,6.000000){\makebox(0,0){\tiny $u_1$}}
\path(68.531695,0.000000)(108.531695,0.000000)
\put(108.531695,0.000000){\circle*{5}}
\path(108.531695,0.000000)(80.000000,20.729490)
\put(80.000000,20.729490){\circle*{5}}
\path(80.000000,20.729490)(90.898138,54.270510)
\put(90.898138,54.270510){\circle*{5}}
\path(90.898138,54.270510)(126.165253,54.270510)
\put(126.165253,54.270510){\circle*{5}}
\path(126.165253,54.270510)(137.063391,20.729490)
\put(137.063391,20.729490){\circle*{5}}
\path(137.063391,20.729490)(108.531695,0.000000)
\put(108.531695,0.000000){\circle*{5}}
\put(108.531695,11.000000){\makebox(0,0){\tiny $X_{2,1}$}}
\put(90.461622,24.128677){\makebox(0,0){\tiny $X_{2,2}$}}
\put(97.363776,45.371323){\makebox(0,0){\tiny $X_{2,3}$}}
\put(119.699615,45.371323){\makebox(0,0){\tiny $X_{2,4}$}}
\put(126.601769,24.128677){\makebox(0,0){\tiny $X_{2,5}$}}
\path(108.531695,0.000000)(148.531695,0.000000)
\put(148.531695,0.000000){\circle*{5}}
\put(148.531695,6.000000){\makebox(0,0){\tiny $u_2$}}
\path(148.531695,0.000000)(188.531695,0.000000)
\put(188.531695,0.000000){\circle*{5}}
\path(188.531695,0.000000)(160.000000,20.729490)
\put(160.000000,20.729490){\circle*{5}}
\path(160.000000,20.729490)(170.898138,54.270510)
\put(170.898138,54.270510){\circle*{5}}
\path(170.898138,54.270510)(206.165253,54.270510)
\put(206.165253,54.270510){\circle*{5}}
\path(206.165253,54.270510)(217.063391,20.729490)
\put(217.063391,20.729490){\circle*{5}}
\path(217.063391,20.729490)(188.531695,0.000000)
\put(188.531695,0.000000){\circle*{5}}
\put(188.531695,11.000000){\makebox(0,0){\tiny $X_{3,1}$}}
\put(170.461622,24.128677){\makebox(0,0){\tiny $X_{3,2}$}}
\put(177.363776,45.371323){\makebox(0,0){\tiny $X_{3,3}$}}
\put(199.699615,45.371323){\makebox(0,0){\tiny $X_{3,4}$}}
\put(206.601769,24.128677){\makebox(0,0){\tiny $X_{3,5}$}}
\path(188.531695,0.000000)(228.531695,0.000000)
\put(228.531695,0.000000){\circle*{5}}
\put(228.531695,6.000000){\makebox(0,0){\tiny $u_3$}}
\path(228.531695,0.000000)(268.531695,0.000000)
\put(268.531695,0.000000){\circle*{5}}
\path(268.531695,0.000000)(240.000000,20.729490)
\put(240.000000,20.729490){\circle*{5}}
\path(240.000000,20.729490)(250.898138,54.270510)
\put(250.898138,54.270510){\circle*{5}}
\path(250.898138,54.270510)(286.165253,54.270510)
\put(286.165253,54.270510){\circle*{5}}
\path(286.165253,54.270510)(297.063391,20.729490)
\put(297.063391,20.729490){\circle*{5}}
\path(297.063391,20.729490)(268.531695,0.000000)
\put(268.531695,0.000000){\circle*{5}}
\put(268.531695,11.000000){\makebox(0,0){\tiny $X_{4,1}$}}
\put(250.461622,24.128677){\makebox(0,0){\tiny $X_{4,2}$}}
\put(257.363776,45.371323){\makebox(0,0){\tiny $X_{4,3}$}}
\put(279.699615,45.371323){\makebox(0,0){\tiny $X_{4,4}$}}
\put(286.601769,24.128677){\makebox(0,0){\tiny $X_{4,5}$}}
\path(268.531695,0.000000)(308.531695,0.000000)
\put(308.531695,0.000000){\circle*{5}}
\put(308.531695,6.000000){\makebox(0,0){\tiny $u_4$}}
\end{picture}
\end{center}
\caption{The graph $\G_n$ for $n=4$.}
\label{figure-string-of-cycles}
\end{figure}
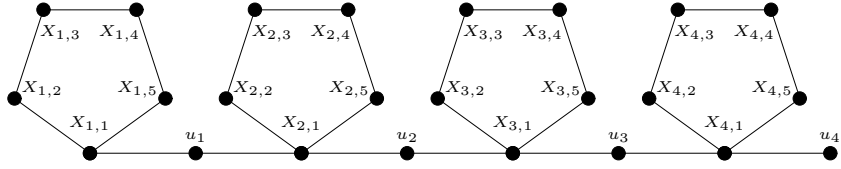

\section{Open problems}
We have determined that according to relative worst order analysis,
\LRU is better than \FIFO on paths and cycles. On some classes
of general access graphs, the two algorithms are incomparable.
It would be interesting to get closer to determining exact
access graphs classes characterizing relationships between the two algorithms.
 We believe that the results for paths and cycles
will form fundamental building blocks in an attack on this problem.
The most obvious class of access graphs to study next is trees.
\LRU  can clearly do better than \FIFO on any tree containing a path of length
$k+1$.
We conjecture that \LRU does at least as well as \FIFO on any tree.
One difficulty in establishing a proof of this is that
for trees, as opposed to the cases of paths and cycles, there exist
worst order sequences for \LRU for which \FIFO performs better than \LRU.

For general access graphs, when showing that \FIFO can do better
than \LRU, we used a family of access graphs, the size of which
grew with the length of the input sequence. It would be interesting
to know if this is necessary, or if such a separation result
can be established on a single access graph of bounded size.

\bibliographystyle{plain}
\bibliography{refs}

\end{document}